\documentclass[11pt]{article}

\usepackage{times}
\usepackage{fullpage,amsmath,amsthm}
\usepackage{amssymb,latexsym,graphicx}
\usepackage{hyperref}

\newenvironment{step}
  {
    \begin{enumerate}

  }
  {\end{enumerate}}

\newenvironment{algorithm*}[1]
  {
    \begin{center}
      \hrulefill\\
      \textbf{#1}
  }
  {
    \vspace{-1\baselineskip}
    \hrulefill
    \end{center}
  }

\newenvironment{protocol*}[1]
  {
    \begin{center}
      \hrulefill\\
      \textbf{#1}
  }
  {
    \vspace{-1\baselineskip}
    \hrulefill
    \end{center}
  }

\newtheoremstyle{game}
  {}{}{\setlength{\parskip}{4pt}\setlength{\parindent}{0pt}}{}{\bfseries}{.}{\newline}
  {\thmname{#1}\thmnumber{~#2}\thmnote{: #3}}
\newtheorem{theorem}{Theorem}

\newtheorem{lemma}[theorem]{Lemma}

\newtheorem{claim}[theorem]{Claim}

\newtheorem{definition}[theorem]{Definition}
\theoremstyle{game}

\newcommand{\CHSH}{\text{CHSH}}
\newcommand{\BAD}{\text{BAD}}
\newcommand{\boxA}{\mathcal{A}}
\newcommand{\boxB}{\mathcal{B}}

\newcommand{\beq}{\begin{equation}}
\newcommand{\eeq}{\end{equation}}
\newcommand{\ket}[1]{|#1\rangle}
\newcommand{\bra}[1]{\langle#1|}
\newcommand{\Tr}{\mbox{\rm Tr}}

\newcommand{\Es}[1]{\textsc{E}_{#1}}
\newcommand{\Exs}[2]{\textsc{E}_{#1}\left[ #2 \right]}
\newcommand{\eps}{\varepsilon}

\DeclareMathOperator{\poly}{poly}

\newcommand{\footremember}[2]{%
   \footnote{#2}
    \newcounter{#1}
    \setcounter{#1}{\value{footnote}}%
}
\newcommand{\footrecall}[1]{%
    \footnotemark[\value{#1}]%
}

\begin{document}

\title{Certifiable Quantum Dice\\[2mm] \Large Or, testable exponential randomness expansion }
\author{Umesh Vazirani\footremember{1}{Computer Science division, UC Berkeley, USA. Supported by ARO Grant W911NF-09-1-0440 and NSF Grant CCF-0905626.}\and Thomas Vidick\footrecall{1}}
\date{}
\maketitle

\begin{abstract}
We introduce a protocol through which a pair of quantum mechanical devices may be used to generate $n$ bits of true randomness from a seed of $O(\log n)$ uniform bits. The bits generated are certifiably random based only on a simple  statistical test that can be performed by the user, and on the assumption that the devices obey the no-signaling principle. No other assumptions are placed on the devices' inner workings. A modified protocol uses a seed of $O(\log^3 n)$ uniformly random bits to generate $n$ bits of true randomness even conditioned on the state of a quantum adversary who may have had prior access to the devices, and may be entangled with them. 
\end{abstract}

\section{Introduction}
A source of independent random bits is a basic resource in many modern-day computational tasks, such 
as cryptography, game theoretic protocols, algorithms and physical simulations. Moreover, these tasks place 
different demands on the quality of the randomness (e.g. the need for privacy in cryptographic applications).
It is of great interest, therefore, to construct a physical device for reliably and provably outputting a stream 
of random bits. Testing such a device poses a fundamental problem ---  since all outputs should be output with equal probabilitythere is no basis for rejecting any particular output of the device. 

Starting in the mid-80's, computer scientists considered the question of extracting truly random bits from 
adversarially controlled physical sources of randomness, such as the semi-random source \cite{SV84}, and weak
random sources \cite{Zuc90}. This sequence of papers has culminated in sophisticated algorithms called randomness
extractors that are guaranteed to output a sequence of truly random bits from physical sources of low-quality
randomness (see~\cite{Sha02} for a survey). It was clear, in a classical world, that these results were the best one could hope for --- while it 
was necessary to assume that the physical device outputs randomness (since that could not be tested), 
minimal assumptions were made about the quality of randomness output. 

\medskip

Quantum mechanics provides a surprising path around this fundamental barrier --- it provides a way of testing 
that the output of a certain kind of device is truly random. Recall the famous CHSH game, illustrated in Figure~\ref{fig:chsh}. In this game two \emph{non-communicating} parties, represented by spatially separated boxes $A$, $B$, are given inputs $x,y\in\{0,1\}$ respectively. Their task is to produce outputs $a,b\in\{0,1\}$ such that the \emph{CHSH condition} $a\oplus b=  x\wedge y$ holds. Let $p_{\texttt{CHSH}}$ be the probability that a certain pair of boxes produces outputs satisfying this condition, when the inputs $x,y$ are chosen uniformly at random. 

\begin{figure}[htb!]
\centering%
\includegraphics[scale=.5, angle = 0]{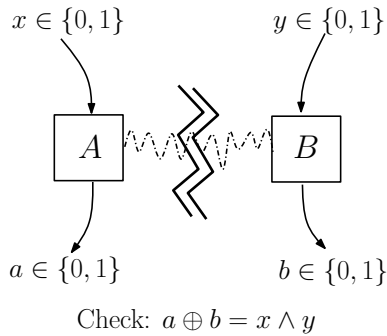}
\caption{The CHSH game. Any pair of boxes $A,B$ is characterized by a distribution $p(a,b|x,y)$ which is required to be \emph{no-signaling}: the marginal distribution of $b$ is independent of $x$, and that of $a$ is independent of $y$.}
\label{fig:chsh}
\end{figure}

Classical players can achieve a success probability at most $p_{\texttt{CHSH}}\leq \frac{3}{4}$, but 
there is a quantum strategy that succeeds with $p_{\texttt{CHSH}} = \cos^2 \pi/8 \approx 0.85$.
Indeed, we may define the \emph{quantum regime} corresponding to success probability 
$3/4<p_{\texttt{CHSH}} \leq \cos^2 \pi/8 \approx 0.85$.
For any value in that range there is a simple quantum-mechanical pair of boxes, still obeying the no-signaling condition, which achieves that success probability.

These well-known facts have a striking consequence: any boxes producing correlations that fall in the quantum regime \emph{must be randomized}! Indeed, deterministic boxes are inherently classical, so that their success probability must fall in the classical regime $p_{\texttt{CHSH}} \leq 3/4$. Hence a simple \emph{statistical test} guaranteeing the presence of randomness, under a single assumption on the process that produced the bits: that it obeys the no-signaling condition. This powerful observation was first made in Colbeck's Ph.D. thesis~\cite{Colbeck09} (see also~\cite{Colbeck11} for an expanded version). The idea was then developed in a paper by Pironio et. al.~\cite{Pironio}, where the first quantitative bounds on the amount of randomness produced were shown. 


\subsection*{An efficient and testable randomness-generation protocol}

This method of generating randomness is not very efficient. Choosing a pair of inputs for the boxes requires $2$ bits of randomness, so the $2$ bits that are output certainly do not contain more randomness than was used.\footnote{In fact, one may show that boxes having a probability of success in the CHSH game that is close to the optimal quantum value produce at most $1.25$ random bits per use, on average~\cite{Pironio}.} 

\begin{figure}
\begin{protocol*}{Protocol~A}
\begin{step}
\item Let $\ell, \Delta$ be two integers given as input. Set $k=\lceil 10\log^2 \ell \rceil$ and $m= \Delta \, \ell$.
\item Choose $T\subseteq [m]$ uniformly at random by selecting each position independently with probability $1/\ell$.
\item Repeat, for $i=1,\ldots,m$:
\begin{step}
\item If $i\notin T$, then 
\begin{step}
\item Set $x=y=0$ and choose $x,y$ as inputs for $k$ consecutive steps. Collect outputs $a,b\in\{0,1\}^k$.
\item If $a\oplus b$ has more than $\lceil 0.16 k \rceil$ $1$'s then reject and abort the protocol. Otherwise, continue.
\end{step}
\item If $i\in T$,
\begin{step}
\item Pick $x,y \in \{0,1\}$ uniformly at random, and set $x,y$ as inputs for $k$ consecutive steps. Collect outputs $a,b\in\{0,1\}^k$.
\item If $a\oplus b$ differs from $x\wedge y$ in more than $\lceil 0.16 k \rceil$ positions then reject and abort the protocol. Otherwise, continue.
\end{step}
\end{step}
\item If all steps accepted, then accept. 
\end{step}
\end{protocol*}
\caption{Protocol~A uses $O(\Delta \log \ell)$ bits of randomness and makes $O(\ell\log^2\ell)$ uses of the boxes. Theorem~\ref{thm:main-class} shows that $\Omega(\ell)$ bits of randomness are produced, with security $\eps = \text{exp}(-\Omega(\Delta))$.}
\label{fig:prot-class}
\end{figure}

Instead, consider the following randomness-efficient protocol. Let $n$ be the target number of random bits to be generated, and $\eps$ a ``security'' parameter. 
Inputs in the protocol are grouped in $m = C\,\lceil  n \log(1/\eps)\rceil$ successive blocks of $k =  10\lceil \log^2 n\rceil$ pairs of inputs each, where $C$ is a large constant. Inputs in a given block consist of a fixed pair $(x,y)$ repeated $k$ times. Most blocks use the $(0,0)$ input, but approximately $10^3\lceil \log (1/\eps)\rceil$ of them are selected at random and marked as ``Bell'' blocks. In those blocks a random pair of inputs $(x,y)\in\{0,1\}^2$ is chosen, and used as inputs throughout the block. Finally, the sequence of outputs produced by the boxes is accepted if, in every block, the CHSH constraint is satisfied by at least $0.84k$ of the blocks's input/output pairs.\footnote{Note that honest boxes, playing each round independently, will indeed satisfy the CHSH condition in each block on average with probability $1-2^{-\Omega(\log^2 n)}$, so that by a union bound it is very unlikely that they will fail the CHSH condition in any of the blocks.} 

\medskip

The following theorem shows that this protocol (formally described as Protocol~A in Figure~\ref{fig:prot-class}) can be used to generate certifiably random bits. 

\begin{theorem}\label{thm:main-class} There exists a constant $C>1$ such that the following holds. Let $\eps>0$ be given, and $n$ an integer. Set $\Delta = 10^3\,\lceil \log(1/\eps)\rceil$, and $\ell = C\,n$. Let $(\boxA,\boxB)$ be an arbitrary pair of no-signaling boxes used to execute Protocol~A, $B$ the random variable describing the bits output by $\boxB$ in protocol~A, and $\CHSH$ the event that the boxes' outputs are accepted in the protocol. Then for all large enough $n$ at least one of the following holds :
\begin{itemize}
\item Either $H_\infty^\eps(B | \CHSH) \geq n,$
\item Or $\Pr\big( \CHSH \big) \leq \eps.$
\end{itemize}
Moreover, Protocol~A requires $O( \log n \log(1/\eps))$ bits of randomness, and makes $O( n \log^2n \log(1/\eps))$ uses of the boxes.  
\end{theorem}

We note that the second condition in the theorem is necessary, as there is always an unavoidable chance that the boxes successfully guess their whole inputs, and deterministically produce matching outputs. The theorem guarantees that the probability of this happening can be bounded by an inverse-exponential in the number of random bits used.

The theorem as stated only guarantees that the bits output by the device have large (smooth) min-entropy. In order to obtain bits that are (close to) uniformly random, one may apply an extractor. There exists efficient constructions of such devices which will convert $B$ into roughly $H_\infty^\eps(B|\CHSH)$ bits that are $\eps$-close, in statistical distance, to uniform. In order to do so, the best extractors will require an additional $O(\log n)$ many uniformly random bits to be used as seed~\cite{GUV07}. 

\medskip

Compared to the basic procedure outlined earlier, Protocol~A  uses two main ideas in order to save on the randomness required. The first idea is to restrict the inputs to  $(0,0)$ most of the time. Only a few randomly placed checks (the Bell blocks) are performed in order to verify that the boxes are generating their inputs honestly. This idea was already used in~\cite{Pironio}, and
led to a protocol with a quadratic $\sqrt{n}\to n$ expansion of randomness. 

The second idea is to systematically group inputs to the boxes into blocks of $k$ successive, identical pairs and check that the CHSH correlations are satisfied on average \emph{in every block}. This is necessary: if one was to only check that the CHSH condition is satisfied on average over the whole protocol, then boxes systematically producing the outputs $(0,0)$ would lead to a large --- close to $100\%$ --- violation. Hence the more robust checking that we perform forces the boxes to play ``honestly'' and produce randomness in essentially every block. 

Moreover, the block structure of the inputs also plays a key role in the analysis of the protocol, 
which is based on the definition of a simple ``guessing game'', explained in Section~\ref{sec:guessing}.
The main point is that if box $\boxB$'s output in a certain block is likely to be a particular string, then Alice, given access to $\boxA$, can guess $\boxB$'s input $y\in\{0,1\}$ based on whether $\boxA$'s output is ``close'' or ``far'' in Hamming distance from that particular string. This provides a way for Alice to guess $\boxB$'s input with probability greater than $1/2$, violating the no-signaling condition placed on the boxes. This style of reasoning can be used to establish that $\boxB$'s output must have high min-entropy, thus yielding Theorem~\ref{thm:main-class}. The proof is given in Section~\ref{sec:class-analysis}.

\medskip

To understand the significance of Theorem~\ref{thm:main-class}, it may be instructive to recall the common paraphrasing
of Einstein's quote from his 1926 letter to Max Born expressing his unhappiness with quantum mechanics as
``God does not play dice with the Universe.'' Clearly a device
based on quantum mechanics can be used to generate randomness --- simply prepare a qubit in the $\ket{0}$ state, 
apply a Hadamard gate, and measure the resulting state in the computational basis: the outcome is a uniformly random bit. 
However, in addition to believing the correctness of quantum mechanics, to trust that such a device produces random bits one 
must believe that the manufacturer is trustworthy, experimentally skilled, and that the device is always
well calibrated. These difficulties are compounded by the fact that the postulates of quantum physics forbid any 
classical observer from getting more than a small probabilistic digest of the internal quantum state of the system.
The randomness generation protocol presented above has the property that the output is guaranteed to be 
random based only on the observed correlations in the output (violations of Bell inequalities), and on the relativistic assumption that information does not travel faster than light. In this sense it might be appropriate to deem that it is ``Einstein certifiable''!

\subsection*{Quantum adversaries}

We have described a simple protocol that guarantees the production of bits that are statistically close to uniform. Suppose these random bits were used later in an interactive cryptographic protocol. In that case it is crucial that the bits generated appear close to uniform not only to the (honest) user of the protocol, but also to any adversary in the cryptographic protocol. 

For concreteness, consider the following catastrophic scenario: the maker of the boxes, call her Eve, inserted an undetectable ``back-door'' by not only entangling $\boxA$ and $\boxB$ together, but extending this entanglement to reach into her own, private, laboratory. Eve knows that the protocol mostly uses $0$'s as inputs to $\boxB$. Betting on this she repeatedly makes a specific measurement on her system, which reliably produces the same output as $\boxB$ in case its input was a $0$. If we assume that $\boxB$'s outputs are uniformly distributed then such a strategy does not obviously violate the no-signaling constraint between $\boxB$ and Eve. But Eve learns most of $\boxB$'s output: while in isolation it may be random, it is totally insecure!

\medskip

We rule out this scenario by showing an analogue to Theorem~\ref{thm:main-class} which also holds in the presence of a quantum adversary. The theorem applies to a slight variant of the protocol used in the previous section, described as Protocol~B in Figure~\ref{fig:prot-quant}. The main differences are that the number of random bits used in that protocol is slightly larger, $O(\log^3 n)$ instead of $O(\log n)$, and the protocol is based on using an ``extended'' version of the CHSH game, which will be introduced in Section~\ref{sec:quant}. 

\begin{theorem}\label{thm:main-quant} Let $\alpha,\gamma>0$ be such that $\gamma \leq 1/(10+8\alpha)$, and $n$ an integer. Set $C = \lceil 100\,\alpha \rceil$, and $\ell = n^{1/\gamma}$. Let $(\boxA,\boxB)$ be an arbitrary pair of no-signaling boxes used to execute Protocol~B, $\CHSH$ the event that the boxes' outputs are accepted in the protocol, and $B'$ the random variable describing the bits output by $\boxB$, conditioned on $\CHSH$. 
Let $E$ be an arbitrary quantum system, possibly entangled with $\boxA$ and $\boxB$, but such that no communication occurs between $\boxA,\boxB$ and $E$ once the protocol starts. 
Then for all large enough $n$ at least one of the following holds:
\begin{itemize}
\item Either $H_\infty^\eps(B' | E) \geq n,$
\item Or $\Pr\big( \CHSH \big) \leq \eps,$
\end{itemize}
where $\eps = n^{-\alpha}$. Moreover, Protocol~B uses only $O(\gamma^{-3}\log^3 n)$ bits of randomness. 
\end{theorem}

Indication that dealing with quantum, rather than classical, adversaries may present substantial new difficulties may be found in the area of strong extractor constructions. There are examples of such constructions, secure against classical adversaries, that dramatically fail in the presence of quantum adversaries with even smaller prior information~\cite{GKKRW}. Luckily, other constructions, such as a very efficient construction due to Trevisan~\cite{Trevisan01}, have been shown secure even against quantum adversaries~\cite{Ta09,DVPR11}. One may use such a ``quantum-proof'' extractor in order to efficiently transform the bits output in Protocol~B into ones that are statistically close to uniform even from the point of view of the adversary at the cost of an additional $O(\log^2 n)$ bits of fresh randomness.  


A reason to think that the power of a quantum adversary in learning $\boxB$'s output may be limited comes from a delicate property of entanglement, its \emph{monogamy}~\cite{Terhal04}. Informally, monogamy states that a tripartite entangled state $\ket{\Psi}_{ABE}$ cannot be maximally entangled both between $A$ and $B$ and between $B$ and $E$. Since Protocol~B enforces very strict correlations between the outputs of $\boxA$ and $\boxB$, one may hope that these correlations will pre-empt any strong correlation between $\boxB$ and an arbitrary $E$. 

\medskip

Interestingly, the proof of Theorem~\ref{thm:main-quant} makes crucial use of the properties of a specific construction of a quantum-proof extractor, based on Trevisan's construction and the $t$-XOR code, that was first outlined in~\cite{DV09}. This construction is used to prove the following information-theoretic lemma. The lemma gives an operational interpretation to a random variable having small smooth min-entropy conditioned on a quantum system, and may be of independent interest. 

\begin{lemma}\label{lem:chain-rule} Let $\rho_{XE}$ be a state such that $X$ is a classical random variable distributed over $m$-bit strings, and $E$ is an arbitrarily correlated quantum system. Let $\eps,\delta>0$, and $K=H_\infty^\eps(X|E)$. Then there exists a subset $V \subseteq [m]$ of size $v=|V| = O(K\log^2 m)$, and for every $v$-bit string $z$ a measurement $M_z$ on $E$ such that, with probability at least $\Omega(\eps^6/m^6)$, $M_{X_V}$ produces a string $Y$ that agrees with $X$ in a fraction at least $1-\frac{1}{\log m}$ of positions. 
\end{lemma}

In essence Lemma~\ref{lem:chain-rule} states that, given access to some of the bits of $X$ (the ones indexed by $V$), and to the quantum system $E$, one can predict the remainder of the string $X$ with inverse-polynomial success probability. In the range of large $K$ (at least inverse-polynomial in $m$), this is much higher than the inverse-exponential probability that one would get by measuring $E$ directly, without using any ``advice'' bits. 

The proof of lemma~\ref{lem:chain-rule} mostly follows from the proof of security of Trevisan's extractor against quantum adversaries presented in~\cite{DVPR11}. Since however it does not follow as a black-box, we give a detailed outline of the proof of the lemma in Appendix~\ref{app:kxor}.

\paragraph{Related work.} Two concurrent and independent papers, the first by Fehr, Gelles and Schaffner~\cite{FGS11} and the second by Pironio and Massar~\cite{PM11} showed the security of a randomness-generation scheme against quantum adversaries in the generic setting in which the violation of \emph{any} Bell inequality is observed. While this approach initially only leads to a polynomial expansion of randomness, both works show that by combining the use of two pairs of devices (that is, four non-communicating boxes in total), one can also obtain a scheme with exponential expansion (in fact, this idea was already suggested in~\cite{Pironio}). The fact that such a composition technique works crucially relies on the original scheme being secure against quantum adversaries.

The guarantees on the amount of randomness, and its security, that are obtained in these works rely on the estimation of the average violation of a Bell inequality throughout a ``generic'' protocol. In contrast, our result is more tailored to the actual protocol we introduce, as well as to the use of the CHSH inequality itself. We see this as a benefit: by providing a simpler, more direct analysis, we hope that our approach may lead to further improvements, and may be more easily adaptable to a variety of  settings. For instance, taking such a direct approach leads us to a protocol achieving exponential expansion with only one device (two boxes) instead of two. The protocol's simplicity contrasts with the relatively involved composition technique that needs to be performed in order to achieve the same expansion in~\cite{FGS11} and~\cite{PM11}.

Recent work by Colbeck and Renner~\cite{CR11} studies a related question, that of improving the quality of a given source of weak randomness. Specifically, they show that if one is given access to a so-called Santha-Vazirani source then one can produce bits that are guaranteed to be statistically close to uniform by using the violation of a specific Bell inequality by a pair of untrusted no-signaling devices.  

\paragraph{Organization of the paper.} We begin with some preliminaries in Section~\ref{sec:prelim}. In Section~\ref{sec:guessing} we introduce the \emph{guessing game}, an important conceptual tool in the proofs of both Theorem~\ref{thm:main-class} and Theorem~\ref{thm:main-quant}. In Section~\ref{sec:class-analysis} we prove Theorem~\ref{thm:main-class}, while Theorem~\ref{thm:main-quant} is proven in Section~\ref{sec:quant}. The proof of Lemma~\ref{lem:chain-rule} mostly follows from known results, and is relegated to Appendix~\ref{app:kxor}. 


\section{Preliminaries}\label{sec:prelim}

\paragraph{Notation.} Given two $n$-bit strings $x,y$ we let $d_H(x,y) = \frac{1}{n}\sum_{i=1}^n |x_i-y_i|$ denote their relative Hamming distance. For $i\in [n]$, we let $x_i$ be the $i$-th bit of $x$, and $x_{<i}$ its $(i-1)$-bit prefix. 

\paragraph{Classical random variables.}
Given a random variable $X\in\{0,1\}^n$, its min-entropy is 
$$H_\infty(X) = -\log \max_{x} \Pr(X=x).$$
 For two distributions $p,q$ on a domain $D$, their statistical distance is 
$$\|p-q\|_1:=(1/2)\sum_{x\in D} |p(x)-q(x)|_1.$$
 This notion of distance can be extended to random variables with the same range in the natural way.  Given $\eps>0$, the smooth min-entropy of a random variable $X$ is 
$$H_\infty^\eps(X) = \sup_{Y,\, \|Y-X\|_1 \leq \eps} H_\infty(Y).$$
 The following simple claim will be useful. 

\begin{claim}\label{claim:smoothcap} Let $\alpha,\eps>0$ and $X$ a random variable such that $H_\infty^\eps(X) \leq \alpha$. Then there exists a set $B$ such that $\Pr(X\in B) \geq \eps$ and for every $x\in B$, it holds that $\Pr(X=x) \geq 2^{-\alpha}$. 
\end{claim}

\begin{proof} Let $B$ be the set of $x$ such that $\Pr(X=x)\geq 2^{-\alpha}$, and suppose $\Pr(X\in B) < \eps$.  Define $Y$ so that $\Pr(Y=x)=\Pr(X=x)$ for every $x\notin B$, $\Pr(Y=x)=0$ for every $x\in B$. In order to normalize $Y$, introduce new values $z$ such that $\Pr(X=z)=0$, and extend $Y$ by defining $\Pr(Y=z) = 2^{-\alpha-1}$ until it is properly normalized. Then $\|Y-X\|_1 < \eps$ and $H_\infty(Y) > \alpha$, contradicting the assumption on the smooth min-entropy of $X$.
\end{proof}

\paragraph{Quantum states.}
Let $X$ be a register containing a classical random variable, which we also call $X$, and $E$ a register containing a quantum state, possibly correlated to $X$. Then the whole system can be described using the cq-state (cq stands for classical-quantum) $\rho_{XE} = \sum_{x} p_X(x)\ket{x}\bra{x} \otimes \rho_x$, where for every $x$ $\rho_x$ is a density matrix, i.e. a positive matrix with trace $1$. Given such a state, the guessing entropy $p_{guess}(X|E)$ is the maximum probability with which one can predict $X$, given access to $E$. Formally, it is defined as 
$$p_{guess}(X|E)_\rho\,=\, \sup_{\{M_x\}} \sum_x p_X(x) \Tr\big(M_x \rho_x \big),$$
where the supremum is taken over all projective operator-valued measurements (POVMs) on $E$.\footnote{A POVM $\{M_x\}$ is given by a set of positive matrices which sum to identity. We refer the reader to the standard textbook~\cite{NC00} for more details on the basics of quantum information theory.} The conditional min-entropy can be defined through the guessing entropy as $H_\infty(X|E)_\rho = -\log p_{guess}(X|E)_\rho$~\cite{KRS09}. We will often omit the subscript $\rho$, when the underlying state is clear. The appropriate distance measure on quantum states is the trace distance, which derives from the trace norm $\|A\|_{tr} = \Tr\big(\sqrt{A^\dagger A}\big)$. This lets us define a notion of smooth conditional min-entropy: $H_\infty^\eps(X|E)_\rho = \sup_{\sigma_{XE},\, \|\sigma_{XE}-\rho_{XE}\|_{tr} \leq \eps} H_\infty(X|E)_\sigma$, where here the supremum is taken over all sub-normalized cq-state $\sigma_{XE}$.  As in the purely classical setting, it is known that this measure of conditional min-entropy is the appropriate one from the point of view of extracting uniform bits~\cite{Ren05}: if $H_\infty^\eps(X|E) = K$ then $K - O(\log 1/\eps)$ bits can be extracted from $X$ that are $\eps$-close to uniform, even from the point of view of $E$. 

\paragraph{The CHSH game.} The following game was originally introduced by Clause, Horne, Shimony and Holt~\cite{Clauser:69a} to demonstrate the non-locality of quantum mechanics. Two collaborating but non-communicating parties, Alice and Bob, are each given a bit $x,y\in \{0,1\}$ distributed uniformly at random. Their goal is to produce bits $a,b$ respectively such that $a\oplus b = x\wedge y$. It is not hard to see that classical parties (possibly using shared randomness) have a maximum success probability of $3/4$ in this game. In contrast, quantum mechanics predicts that the following strategy, which we will sometimes refer to as the ``honest'' strategy, achieves a success probability of $\cos^2(\pi/8) \approx 0.85$. Alice and Bob share an EPR pair $\ket{\Psi} = \frac{1}{\sqrt{2}}\ket{00}+\frac{1}{\sqrt{2}}\ket{11}$. Upon receiving her input, Alice measures either in the computational $(x=0)$ or the Hadamard $(x=1)$ basis. Bob measures in the computational basis rotated by either $\pi/8$ $(y=0)$ or $3\pi/8$ $(y=1)$. One can then verify that, for every pair of inputs $(x,y)$, this strategy produces a pair of correct outputs with probability exactly $\cos^2 (\pi/8)$.


\section{The guessing game}\label{sec:guessing}

Consider the following simple guessing game. In this game, there are two cooperating players, Alice and Bob. At the start of the game Bob receives a single bit $y\in\{0,1\}$ chosen uniformly at random. The players are then allowed to perform arbitrary computations, but are not allowed to communicate. At the end of the game Alice outputs a bit $a$, and the players win if $a=y$. 

Clearly, any strategy with success probability larger than $\frac{1}{2}$ indicates a violation of the no-communication assumption between Alice and Bob. At the heart of the proofs of both Theorem~\ref{thm:main-class} and Theorem~\ref{thm:main-quant} is a reduction to the guessing game. Assuming there existed a pair of boxes violating the conclusions of either theorem, we will show how these boxes may be used to devise a successful strategy in the guessing game, contradicting the no-signaling assumption placed on the boxes.

\medskip

To illustrate the main features of the strategies we will design later, consider the following simplified setting. Let $\boxA,\boxB$ be a given pair of boxes taking inputs $X,Y\in \{0,1\}$ and producing outputs $A,B\in\{0,1\}^k$ respectively. Assume the following two properties hold. First, if the input to $\boxB$ is $Y=0$ then its output $B$ is essentially deterministic, in the sense that $B=b_0$ with high probability. Second, whatever their inputs, the boxes' outputs satisfy the CHSH constraint on average: at least $84\%$ of $i\in [k]$ are such that $A_i\oplus B_i = X\wedge Y$. Then we claim that there is a strategy for Alice and Bob in the guessing game, using $\boxA$ and $\boxB$, that succeeds with probability strictly larger than $1/2$, demonstrating that the boxes must be signaling. 

Alice and Bob's strategy is the following. Alice is given access to $\boxA$ and Bob to $\boxB$. Upon receiving his secret bit $y$, Bob inputs it to $\boxB$, collecting outputs $b\in\{0,1\}^k$. Alice chooses an $x\in\{0,1\}$ uniformly at random, and inputs it to $\boxA$, collecting outputs $a\in\{0,1\}^k$. Let $b_{0}$ be the $k$-bit string with the highest probability of being output by $\boxB$, conditioned on $y=0$. 
Alice makes a decision as follows: she computes the relative Hamming distance $d = d_H(a,b_0)$. If $d < 0.2$ she claims ``Bob's input was $0$''. Otherwise, she claims ``Bob's input was $1$''. 

By assumption, if Bob's secret bit was $y=0$, then his output is almost certainly $b_0$. By the CHSH constraint, independently of her input Alice's output $a$ lies in a Hamming ball of radius $0.16$ around $b_0$. So in this
case she correctly decides to claim ``Bob's input was $0$''.

In the case that Bob's secret bit was $y=1$, the analysis is more interesting. Let $b$ be the actual output 
of $\boxB$. Let $a_0$ and $a_1$ be $\boxA$'s output in the two cases $x=0$ and $x=1$ respectively. 
We claim that the Hamming distance $d_H(a_0,a_1)\geq 0.68$. This is because by the CHSH constraint, 
$d_H(a_0,b) \leq 0.16$, while $d_H(a_1,b) \geq 0.84$. Applying the triangle inequality gives the lower bound
on the distance between $a_0$ and $a_1$. This lower bound is large enough that both $a_0$ and $a_1$
cannot lie in the Hamming ball of radius $0.16$ around $b_0$ (observe that this argument makes no 
use of the actual location of $b$!). Thus in the case $y=1$, Alice correctly outputs ``Bob's input was $1$''
with probability at least $1/2$. 

Overall Alice and Bob succeed in the guessing game with probability $3/4$, which contradicts no-signaling.

\medskip

Clearly there is a lot of slack in the above reasoning, since for contradiction it suffices to succeed in the 
guessing game with any probability strictly greater than $1/2$. By being more careful it is possible to 
allow Bob's output on $y=0$ to have more min entropy, as well as allow for a 
small probability that the boxes' outputs may not satisfy the CHSH constraint:

\begin{lemma}\label{lem:guessing} Let $\beta,\gamma>0$ be such that $\gamma+2\beta < 1/4$, and $k$ an integer. Suppose given a pair of boxes $\boxA,\boxB$, taking inputs $X,Y\in \{0,1\}$ and producing outputs $A,B\in\{0,1\}^k$ each. Suppose the following conditions hold:
\begin{enumerate}
\item When given input $0$, the distribution of outputs of $\boxB$ has low min-entropy: there exists a $b_0\in\{0,1\}^k$ such that $\Pr(B=b_0|Y=0) \geq 1-\gamma$,
\item The boxes' outputs satisfy the CHSH condition, on average: 
$$\Pr\big( \,\#\big\{i\in[k],\,A_i\oplus B_i \neq X\wedge Y\big\}\, >\, 0.16\,k \,\big) \,\leq\, \beta.$$
\end{enumerate}
Then there is a strategy for Alice and Bob, using $\boxA$ and $\boxB$, with gives them success probability strictly greater than $1/2$ in the guessing game.
\end{lemma}

\begin{proof}
Alice and Bob's strategy in the guessing game is as described above. Let $b_{0}$ be the $k$-bit string that is most likely to be output by $\boxB$, conditioned on $y=0$. 

\medskip

We first show that, if Bob's input was $y=0$, then Alice claims that Bob had a $0$ with probability at least $1-\gamma-2\beta$. By the first condition in the lemma, Bob obtains the output $b_{0}$ with probability at least $1-\gamma$. Moreover, by the second condition the CHSH constraint will be satisfied with probability at least $1-2\beta$ on average over Alice's choice of input, given that Bob's input was $y=0$. Given $y=0$, whatever the input to $\boxA$ the CHSH constraint states that $d_H(a,b)<0.16$. Hence by a union bound Alice will obtain an output string $a$ at relative Hamming distance at most $0.16$ from $b_{0}$ with probability at least $1-\gamma-2\beta$.

\medskip

Next we show that, in case Bob's input in the guessing game is $y=1$, Alice claims that Bob had a $1$ with probability at least $\frac{1}{2}\big(1-8\beta)$.
Let $b'$ the actual output produced by Bob.
By the second condition in the lemma and Markov's inequality, with probability at least $1-4\beta$ the output $b'$ is such that the CHSH constraint will be satisfied with probability at least $1-4\beta$ simultaneously for both of Alice's possible choices of input. 

Suppose this holds. If Alice chooses $x=0$ then the CHSH constraint indicates that the corresponding $a_{0}$ should be such that $d_H(a_{0},b')\leq 0.16$, while in case she chooses $x=1$ her output $a_{1}$ should satisfy $d_H(a_{1},b') \geq 0.84$. By the triangle inequality, $d_H(a_{0},a_{1}) \geq 0.68$: whatever the value of $b'$, only one of $a_{0}$ or $a_{1}$ can be at distance less than $0.2$ from $b_0$. 
By a union bound, with probability at least $1-8\beta$ there is a choice of input for Alice that will make her claim Bob had a $1$, and she chooses that input with probability $1/2$. 

\medskip

The two bounds proven above together show that Alice's probability of correctly guessing Bob's input in the guessing game is at least
$$p_{succ}\,\geq\, \frac{1}{2}\big(1-2\gamma\big) + \frac{1}{2} \frac{1-8\beta}{2} \,=\, \frac{1}{2} + \Big( \frac{1}{4}-2\beta-\gamma\Big),$$
which is greater than $1/2$ whenever $2\beta+\gamma < 1/4$, proving Lemma~\ref{lem:guessing}.  
\end{proof}


\section{Proof of the main result}\label{sec:class-analysis}

Theorem~\ref{thm:main-class} asserts that, given any pair $(\boxA,\boxB)$ of non-signaling boxes, if the outputs of $\boxB$ do not contain much min-entropy (when its inputs are chosen as in Protocol~A, described in Figure~\ref{fig:prot-class}), then the boxes can only satisfy the CHSH constraints imposed in the protocol with small probability. 

We prove Theorem~\ref{thm:main-class} by a reduction to the guessing game introduced in Section~\ref{sec:guessing}. Suppose that there existed a pair of boxes such that neither of the theorem's conclusions was satisfied. Recall that Protocol~A calls for a total of $mk$ uses of the boxes, divided into $m$ blocks of $k$ pairs of identical inputs each. 
 We show that, provided the $\CHSH$ constraints are satisfied in all blocks with non-negligible probability, there must exist a special block $i_0\in [m]$ in which the boxes' outputs, conditioned on specific past values, have properties close to those required in Lemma~\ref{lem:guessing}. This lets us carry out a reduction to the guessing game, leading to a contradiction of the no-signaling assumption. 
The exact properties of the special block that we obtain are described in Claim~\ref{claim:special-it} below.

\paragraph{Modeling events in the protocol.} 
To model the situation, we introduce four sequences of random variables $X=(X_i),Y=(Y_i),A=(A_i),B=(B_i) \in \big(\{0,1\}^{k}\big)^m$, where $m$ is the number of blocks of the protocol. $X$ and $Y$ are distributed as in Protocol~A, and $A,B$ are random variables describing the boxes' respective outputs when their inputs are $X$ and $Y$. 
For $i\in [m]$, let $\CHSH_i$ be the event that $d_H(A_i\oplus B_i , X_i \wedge Y_i) \leq 0.16$, and $\CHSH = \bigwedge_i \CHSH_i$.  We will also use the shorthand $\CHSH_{< i } = \bigwedge_{j<i} \CHSH_j$. Finally, we let $T_j$ be a random variable denoting the $j$-th Bell block, chosen jointly by Alice and Bob at the start of Protocol~A. 

\begin{claim}\label{claim:special-it} There exists a constant $C>1$ such that the following holds. Let $2^{-Cn}<\eps<1/5$ and $\Delta=10^3\lceil \log (1/\eps)\rceil$. Suppose that (i) $H_\infty^\eps(B|CHSH) \leq n$, and (ii)  $\Pr(\CHSH) \geq \eps$. Let $m=C\,\Delta n$. Then for all large enough $n$ there exists an index $j_0$ and a set $G$ satisfying $\Pr(G)\geq \eps^5$ such that the following hold.  
\begin{itemize} 
\item $\boxB$'s output in the $j_0$-th Bell block ${T_{j_0}}$ is essentially deterministic: 
\beq\label{eq:spec-cond2}
\forall b\in G,\qquad \Pr(B_{T_{j_0}} = b_{T_{j_0}} | \CHSH_{<{{T_{j_0}}}}, B_{<{T_{j_0}}} = b_{<{T_{j_0}}})\geq 0.99,
\eeq
\item The CHSH condition is satisfied with high probability in the $j_0$-th Bell block ${T_{j_0}}$: 
\beq\label{eq:spec-cond1a}
\forall b\in G,\qquad\Pr(\CHSH_{T_{j_0}} | \CHSH_{<{T_{j_0}}},B_{< {T_{j_0}}} = b_{< {T_{j_0}}}) \geq 0.9.
\eeq
\end{itemize}
\end{claim}

The proof of Claim~\ref{claim:special-it} mostly follows from an appropriate chained application of Baye's rule, and is given in Appendix~\ref{app:special}. In order to conclude the proof of Theorem~\ref{thm:main-class} it remains to show  how the special block identified in Claim~\ref{claim:special-it} can be used to show that boxes $\boxA$ and $\boxB$ satisfying the claim's assumptions may be used successfully in the guessing game. 

Consider the following strategy for Alice and Bob in the guessing game. In a preparatory phase (before Bob receives his secret bit $y$), Alice and Bob run protocol~A with the boxes $\boxA$ and $\boxB$, up to the $i_0$-th block (excluded). Bob communicates $\boxB$'s  outputs up till that block to Alice. Together they check that the CHSH constraint is satisfied in all blocks preceding the $i_0$-th; if not they abort. They also verify that Bob's outputs are the prefix of a string $b\in G$; if not they abort. The guessing game can now start: Alice and Bob are separated and Bob is given his secret input $y$. 

Given the conditioning that Alice and Bob have performed before the game started, once it starts boxes $\boxA$ and $\boxB$ can be seen to satisfy both conditions of Lemma~\ref{lem:guessing}. Indeed, since under the input distribution specified in Protocol~A $\boxB$ receives a $0$ as input in block $i_0$ with probability at least $1/2$, condition 1. in Lemma~\ref{lem:guessing} holds with $\gamma = 1/50$ as a consequence of item~1 in Claim~\ref{claim:special-it}. Condition~2 in Lemma~\ref{lem:guessing} puts a bound on the probability of the $\CHSH$ condition being satisfied under the uniform input distribution. Given that in Protocol~A inputs in a Bell block are chosen according to the uniform distribution as well, item~2 from Claim~\ref{claim:special-it} implies that condition~2 holds with $\beta = 1/10$. Since $\gamma+2\beta = 0.22 < 1/4$, Lemma~\ref{lem:guessing} concludes that the boxes~$\boxA$ and $\boxB$ must be signaling in the $i_0$-th block, a contradiction. This finishes the proof of Theorem~\ref{thm:main-class}.


\section{Producing random bits secure in the presence of a quantum adversary}\label{sec:quant}

In this section we prove Theorem~\ref{thm:main-quant}. We first give an overview of the proof, describing the main steps, in the next section. The formal proof is given in Section~\ref{sec:proof-quant}

\subsection{Proof overview}
Theorem~\ref{thm:main-quant} is based on Protocol~B, a variant of Protocol~A which replaces the use of the CHSH game by the following ``extended'' variant.  In this game each box may receive one of four possible inputs, labeled $(A,0),(A,1),(B,0),(B,1)$. An input such as ``$(A,1)$'' to either box means: ``perform the measurement that $\boxA$ would have performed in the honest CHSH strategy, in case its input had been a $1$''. The advantage of working with this game is that there exists an optimal strategy (the one directly derived from the honest CHSH strategy) in which both players always output identical answers when their inputs are equal. 

\begin{figure}
\begin{protocol*}{Protocol~B}
\begin{step}
\item Let $\ell, C$ be two integers given as input. Set $k=\lceil 10\log^2 \ell \rceil$ and $m= \lceil C\ell\log^2 \ell \rceil$.
\item Choose $T\subseteq [m]$ uniformly at random by selecting each position independently with probability $1/\ell$.
\item Repeat, for $i=1,\ldots,m$:
\begin{step}
\item If $i\notin T$, then 
\begin{step}
\item Set $x=y=(A,0)$ and choose $x,y$ as inputs for $k$ consecutive steps. Collect outputs $a,b\in\{0,1\}^k$.
\item If $a\neq b$ then reject and abort the protocol. Otherwise, continue.
\end{step}
\item If $i \in T$,
\begin{step}
\item Pick $x\in\{(A,0),(A,1)\}$ and $y\in\{(A,0),(B,0)\}$ uniformly at random, and set $x,y$ as inputs for $k$ consecutive steps. Collect outputs $a,b\in\{0,1\}^k$.
\item If either $a= b$ and $x = y$, or $d_H(a,b) \leq 0.16$ and $y = (B,0)$, or $d_H(a,b) \in [0.49,0.51]$ and $x =(A,1)$ and $y = (A,0)$ then continue. Otherwise reject and abort the protocol. 
\end{step}
\end{step}
\item If all steps accepted, then accept. 
\end{step}
\end{protocol*}
\caption{Protocol~B uses $O(\log^3 \ell)$ bits of randomness and makes $O(\ell\log^4\ell)$ uses of the boxes. Theorem~\ref{thm:main-quant} shows that $\Omega(\ell^{\gamma})$ bits of randomness are produced, where $\gamma>0$ is a constant depending on the security parameter $\eps$ one wants to achieve.}
\label{fig:prot-quant}
\end{figure}

Protocol~B follows the same structure as Protocol~A. Inputs are divided into groups of $k = \lceil 10 \log^2 n\rceil$ identical inputs. There are $m = O(n^{1/\delta}\log^2 n)$ successive blocks, where $\delta>0$ is a small parameter. Most blocks use the same input $(A,0)$ to both boxes. A random subset $T\subseteq[m]$ of approximately $\log^2 n$ blocks are designated as Bell blocks. In such blocks $\boxA$ is given an input at random in $\{(A,0),(A,1)\}$, while $\boxB$ is given an input at random in $\{(A,0),(B,0)\}$. 

\medskip

As in the proof of Theorem~\ref{thm:main-class} we will prove Theorem~\ref{thm:main-quant} by contradiction, through a reduction to the guessing game. In the non-adversarial case the crux of the reduction consisted in identifying a special block $i_0 \in [m]$ in which $\boxB$'s output $B$ was essentially deterministic, conditioned on past outputs. In the adversarial setting, however, $B$ may be perfectly uniform, and such a block may not exist. Instead, we start by assuming for contradiction that the min-entropy of Bob's output conditioned on Eve's information is small: $H_\infty^\eps(B|E) \leq n$. 

Previously in the guessing game Alice tried to guess Bob's secret input $y\in \{0,1\}$. She did so by using her prediction for $\boxB$'s outputs, together with the CHSH constraint and her own box $\boxA$'s outputs. Here we team up Alice and Eve. Alice will provide Eve with some information she obtained in previous blocks of the protocol, and based on that information Eve will attempt to make an accurate prediction for $\boxB$'s outputs in the special block. Alice will then use that prediction to guess $y$, using as before the CHSH constraint and her own box $\boxA$'s outputs. 

\paragraph{The reconstruction paradigm.}
We would like to show that, under our assumption on $H_\infty^\eps(B|E)$, Eve can perform the following task: accurately predict (part of) $B$, given auxiliary information provided by Alice. We accomplish this by using the ``reconstruction'' property of certain extractor constructions originally introduced by Trevisan~\cite{Trevisan01}. Recall that an extractor is a function which maps a string $B$ with large min-entropy (conditioned on side information contained in $E$) to a (shorter) string $Z$ that is statistically close to uniform even from the point of view of an adversary holding $E$. The reconstruction proof technique proceeds as follows: Suppose an adversary breaks the extractor. Then there exists another adversary who, given a small subset of the bits of the extractor's input as ``advice'', can reconstruct the \emph{whole} input. Hence the input's entropy must have been at most the number of advice bits given. 

For the purposes of constructing extractors, one would then take the contrapositive to conclude that, provided the input has large enough entropy, the extractor's output must be indistinguishable from uniform, thereby proving security. Here we work \emph{directly} with the reconstruction procedure. Suppose that $B$ has low min-entropy, conditioned on Eve's side information. If we were to apply an extractor to $B$ in order to extract \emph{more} bits than its conditional min-entropy, then certainly the output would not be secure: Eve would be able to distinguish it from a uniformly random string. The reconstruction paradigm states that, as a consequence, there is a strategy for Eve that successfully predicts the \emph{entire} string $B$, given a subset of its bits as advice --- exactly what is needed from Eve to facilitate Alice's task in the guessing game. 

\paragraph{The $t$-XOR extractor.} At this stage we are faced with two difficulties. The first is that the reconstruction paradigm was developed in the context of classical adversaries, who can repeat predictive measurements at will. Quantum information is more delicate, and may be  modified by the act of measuring. The second has to do with the role of the advice bits: since they come from $\boxB$'s output $B$ we need to ensure that, in the guessing game, Alice can indeed provide this auxiliary information to Eve, {\em without} communicating with Bob. 

In order to solve both problems we focus on a specific extractor construction, the $t$-XOR extractor $E_t$ (here $t$ is an integer such that $t = O(\log^2 n)$). For our purposes it will suffice to think of $E_t$ as mapping the $mk$-bit string $B$ to a string of $r\ll n$ bits, each of which is the parity of a certain subset of $t$ out of $B$'s $mk$ bits. Which parities is dictated by an extra argument to the extractor, its seed, based on the use of combinatorial designs. Formally,
\begin{align*} 
E_t\,:\, \{0,1\}^{mk} \times \{0,1\}^s &\to \,\{0,1\}^r\\
(b,y)\qquad &\mapsto \,\big( C_t^1(b,y),\ldots, C_t^r(b,y) \big),
\end{align*}
where $C_t^i(b,y)$ is the parity of a specific subset of $t$ bits of $x$, depending on both $i$ and $y$. 

Suppose that Eve can distinguish the output of the extractor $Z=E_t(B,Y)$ from a uniformly random string with success probability $\eps$. In the first step of the reconstruction proof, a hybrid argument is used to show that Eve can predict the parity of $t$ bits of $B$ chosen at random with success $\eps/r$, given access to the parities of $O(r)$ other subsets of $t$ bits of $B$ as advice. This step uses specific properties of the combinatorial designs.

The next step is the most critical. One would like to argue that, since Eve can predict the parity of a random subset of $t$ of $B$'s bits, she can recover a string that agrees with \emph{most} of the $t$-XORs of $B$. One could then appeal to the approximate list-decoding properties of the $t$-XOR code in order to conclude that Eve may deduce a list of guesses for the string $B$ itself. Since, however, Eve is quantum, the fact that she has a measurement predicting \emph{any} $t$-XOR does not imply she has one predicting \emph{every} $t$-XOR: measurements are destructive and distinct measurements need not be compatible. This is a fundamental difficulty, which arises e.g. in the analysis of random access codes~\cite{ANTV02}. To overcome it one has to appeal to a subtle argument due to Koenig and Terhal~\cite{KT07}. They show that without loss of generality one may assume that Eve's measurement has a specific form, called the \emph{pretty-good measurement}. One can then argue that this specific measurement may be refined into one that predicts a guess for the whole list of $t$-XORs of $B$, from which a guess for $B$ can be deduced by list-decoding the $t$-XOR code. 

The security of the $t$-XOR extractor against quantum adversaries 
was first shown by Ta-Shma~\cite{Ta09}, and later improved in~\cite{DV09,DVPR11}. As such, the argument above is not new. Rather, our contribution is to observe that it proves \emph{more} than just the extractor's security. Indeed, summarizing the discussion so far we have shown that, if $H_\infty^\eps(B|E) \leq n$, then there is a measurement on $E$ which, given a small amount of information about $B$ as advice, reconstructs a good approximation to the \emph{whole} string $B$ with success probability $\poly(\eps/r)$. (This is essentially the statement that is made in Lemma~\ref{lem:chain-rule}.) Most crucially, the bits of information required as advice are localized to a small subset of bits of $B$, of the order of the number of bits of information Eve initially has about that string. This property holds thanks to the specific extractor we are using, which is \emph{local}: every bit of the output only depends on few bits of the input.

\paragraph{Completing the reduction to the guessing game.}
In the guessing game it is Alice who needs to hand the advice bits to Eve. Indeed, if Bob, holding box $\boxB$, was to hand them over, they could leak information about his secret input $y$: some of the advice bits may fall in blocks of the protocol that occur \emph{after} the special block $i_0$ in which Bob is planning to use his secret $y$ as input. This leak of information defeats the purpose of the guessing game, which is to demonstrate signaling between $\boxA$ and $\boxB$. 

Hence the ``extended'' variant of the CHSH game introduced in Protocol~B: since in most blocks the inputs to both $\boxA$ and $\boxB$ are identical, by the extended CHSH constraint enforced in the protocol their outputs should be identical. The relatively few advice bits needed by Eve occupy a fixed set of positions, and with good probability all Bell blocks will fall outside of these positions, in which case Alice can obtain the advice bits required by Eve directly from $\boxA$'s outputs. 

The proof of Theorem~\ref{thm:main-quant} is now almost complete, and one may argue as in Lemma~\ref{lem:guessing} that Alice and Eve together will be able to successfully predict Bob's secret input in the guessing game, contradicting the no-signaling assumption placed on $\boxA$ and $\boxB$. A more detailed proof of the theorem is given in the next section.

\subsection{Proof of Theorem~\ref{thm:main-quant}}\label{sec:proof-quant}

We proceed to formally prove Theorem~\ref{thm:main-quant}, using Lemma~\ref{lem:chain-rule} to perform a reduction to the guessing game (Lemma~\ref{lem:chain-rule} is proved in Appendix~\ref{app:kxor}). Protocol~B is described in Figure~\ref{fig:prot-quant}. It consists of $m = \lceil C\ell\log^2 \ell\rceil$ blocks of $k=\lceil 10\log^2 n\rceil$ rounds each, where $C$ is a large constant, $\ell = n^{1/\gamma}$ and $n$ is the target amount of min-entropy. Each round of the protocol selects inputs to the boxes coming from the ``extended CHSH'' game. That game has four questions per party: $(A,0),(A,1),(B,0),(B,1)$. We expect honest boxes to apply the following strategy. They share a single EPR pair, and perform the same measurement if provided the same input. On input $(A,0)$ the measurement is in the computational basis, and on input $(A,1)$ it is in the Hadamard basis $\{\ket{+},\ket{-}\}$, with the outcome $\ket{+}$ being associated with the output '$0$'. On input $(B,0)$ the measurement is in the basis $\{\cos^2 (\pi/8) \ket{0} + \sin^2 (\pi/8) \ket{1}, \sin^2 (\pi/8) \ket{0} - \cos^2 (\pi/8) \ket{1}\}$, with the first vector being associated with the outcome '$0$'. 

\paragraph{Modeling.} To model the situation, introduce four sequences of random variables $X=(X_i),Y=(Y_i),A=(A_i),B=(B_i) \in \big(\{0,1\}^{k}\big)^m$. $X$ and $Y$ are distributed as in protocol~B, while $A,B$ are random variables describing the boxes' respective outputs when their inputs are $X$ and $Y$. 
For $i\in [m]$, let $\CHSH_i$ be the following event:
$$ \CHSH_i \,=\, \begin{cases} A_i = B_i & \text{ if } X_i = Y_i,\\ d_H(A_i,B_i) \leq 0.16 & \text{ if } Y_i = (B,0),\\  d_H(A_i,B_i) \in [0.49,0.51] &\text{ if } X_i = (A,1)\text{ and } Y_i = (A,0). \end{cases}$$
Honest CHSH boxes as described above satisfy $\CHSH_i$ with probability $1-2^{-\Omega(k)}$. Let $\CHSH = \bigwedge_i \CHSH_i$. 

We introduce two new random variables to model the adversary Eve's behavior, when she performs the measurement promised by Lemma~\ref{lem:chain-rule}. We use $E^A = (E^A_i) \in \big( \{0,1\}^k \big)^m$ to denote the outcome of that measurement when the required advice bits are the bits $A_V$ taken from $\boxA$'s outputs, and $E^B = (E^B_i) \in \big( \{0,1\}^k \big)^m$ to denote its outcome when they are the bits $B_V$ taken from $\boxB$'s output (here $V$ is a fixed subset of $[km]$ that will be specified later). Let $G^A$ be the event that $d_H(E^A,B) < f_e$, and $G^B$ the event that $d_H(E^B,B) < f_e$, where $f_e>0$ is a parameter to be specified later. Let $j\in T$ be an index that runs over the blocks that have been designated as Bell blocks in the protocol (where $T$ itself is a random variable).
Given a Bell block $j$, let $G^A_j$ be a boolean random variable such that $G^A_j= 1$ if and only if either $d_H( E^A_j, B ) < 0.01$ and $Y_j = (A,0)$, or $d_H(E^A_j,B) < 0.17$ and $Y_j = (B,0)$. Define $G^B_j$ symmetrically with respect to $E^B$ instead of $E^A$. 

\medskip

We prove Theorem~\ref{thm:main-quant} by contradiction. Assume that both the theorem's conclusions are violated, so that (i)  $H_\infty^\eps(B' | E) \geq n,$ where $B'$ is a random variable describing the distribution of $\boxB$'s outputs conditioned on $\CHSH$, and $\Pr\big( \CHSH \big) \leq \eps$. Here $\eps = n^{-\alpha}$, where $\alpha>0$ is a parameter. 

\medskip

The first step is to apply Lemma~\ref{lem:chain-rule} with $X=B'$. The conclusion of the lemma is that there exists a subset $V\subseteq [km]$ of size $|V| = O(m^{\gamma}\log^2 m)$ such that, letting $f_e = 1/(\log mk)$, we have $p_s := \Pr(G^B | \CHSH ) = \Omega(\eps^7/n^6) = \Omega( n^{-7(\alpha+\gamma)})$. 

$G^B$ denotes the event that Eve correctly predicts $B$ on a fraction at least $1-f_e$ of positions.
 Since in Protocol~B the Bell blocks form only a very small fraction of the total, a priori it could still be that Eve's prediction is systematically wrong on all Bell blocks, preventing us from successfully using them in the guessing game.

The following claim shows Eve's errors cannot be concentrated in the Bell blocks. The intuition is the following. If $\boxB$'s input in a Bell block is $(A,0)$ then nothing distinguishes this block from most others, so that Eve's prediction has no reason of being less correct than average. However, blocks in which its input is $(B,0)$ are distinguished. We rule out the possibility that Eve's errors are concentrated in such blocks by appealing to the no-signaling condition between Eve and $\boxA$. Indeed, about half of Bell blocks in which $\boxB$'s input is $(B,0)$ are such that $\boxA$'s input for the same block is $(A,0)$: looking only at $\boxA$'s inputs they are indistinguishable from most other blocks. We will argue that, as long as the CHSH constraint is satisfied, Eve might as well have been given the advice bits by Alice, in which case there is no reason for her to make more errors than average in those blocks. 

\begin{claim}\label{claim:eve-pred} Let $T$ be the set of Bell blocks selected in Protocol~B. Then there exists a constant $c_e < 10^{-3}$ such that the following holds. 
$$
\Pr\Big( \Es{j\in T} \big[\,G^A_j \,\big]\,>\, 1-\frac{c_e}{\log n} ,\, \CHSH\Big) \, = \, \Omega(p_s\eps)\,=\,\Omega\big(n^{-8(\alpha+\gamma)}\big).
$$
\end{claim}

The proof of Claim~\ref{claim:eve-pred} is given in Appendix~\ref{app:special}. Based on this claim we can show an analogue of Claim~\ref{claim:special-it} which will let us complete the reduction to the guessing game. Claim~\ref{claim:eve-pred} shows that with probability $\Omega(p_s \eps)$ Eve's prediction will be correct on a fraction at least $1-c_e/\log n$ of Bell blocks. Since there are $O(\log^2 n)$ such blocks in Protocol~B, with the same probability Eve only makes errors on a total number $w_e = O(\log n)$ of Bell blocks. Group the Bell blocks in groups of $20w_e$ successive blocks, and let $k$ be an index that runs over such groups; there are $O(\log n)$ of them. Let $G^A_k$ be the event that Eve's prediction is correct in at least $99\%$ of the Bell blocks in group $k$: $G^A_k = 1$ if and only if $\Es{j\sim k} G^A_j \geq 0.99$, where the average is taken over the Bell blocks comprising group $k$. By Markov's inequality, it follows from Claim~\ref{claim:eve-pred} that $\Pr(\wedge_k G^A_k,\,\CHSH) = \Omega(p_s \eps)$. 

\begin{claim}\label{claim:special-quant} For all large enough $n$ there exists a Bell block $j_0\in T$ such that, in that block, it is highly likely that both Eve's prediction (when given advice bits from $\boxA$'s output) is correct and the CHSH constraint is satisfied, conditioned on this being so in past iterations:
\beq\label{eq:spec-quant-1}
\Pr( G^A_{j_0}, CHSH_{j_0} | CHSH_{j<{j_0}}, G^A_{k<k_0}) \geq 0.98,
\eeq
where $k_0$ is the index of the group containing the $j_0$-th Bell block. 
\end{claim}

\begin{proof}
 By the chain rule, since there are $O(\log n)$ groups there will exist a group $k_0$ in which Eve's prediction is correct, and the CHSH condition is satisfied, with probability at least $0.99$, when conditioned on the same holding of all previous groups. Since by definition Eve being correct in the group means that she is correct in $99\%$ of that group's blocks, there is a specific block $j_0$ in which she is correct with probability at least $0.98$. 
\end{proof}

The reduction to the guessing game should now be clear, and follows along the same lines as the proof of Theorem~\ref{thm:main-class} given in Section~\ref{sec:class-analysis}. Alice and Bob run protocol~B, including the selection of all Bell blocks $T$, with the boxes $\boxA$ and $\boxB$, up to the $j_0$-th Bell block (excluded). Bob communicates $\boxB$'s  outputs up till that block to Alice. They check that the CHSH constraint is satisfied in all blocks previous to the $j_0$-th; if not they abort. The guessing game can now start: Alice and Bob are separated and Bob is given his secret input $y$. If $y=0$ then he chooses $(A,0)$ as input to $\boxB$ in the $j_0$-th block; otherwise he chooses $(B,0)$. He then completes the protocol honestly. Alice chooses an input $x\in\{(A,0),(A,1)\}$ at random for the $j_0$-th block, and then completes the protocol honestly. 

In order to help her guess Bob's input, Alice has access to the eavesdropper Eve. Alice gives the bits $a_V$ taken from $\boxA$'s output string $a$ as advice bits to Eve. Eve makes a prediction $e$ for Bob's output. Alice checks that the event $G^A_{<k_0}$ is satisfied. If not she aborts. If so, by Claim~\ref{claim:special-quant} we know that both $\CHSH_{j_0}$ and $G^A_{j_0}$ are satisfied with probability at least $0.98$, so this must be so with probability at least $0.92$ for each of the four possible pair of inputs $(x,y)$ given to $\boxA$ and $\boxB$ in the $j_0$-th block. 

\medskip

Alice makes her prediction as follows: if either $\boxA$'s input was $(A,0)$ and its output agrees with Eve's prediction on at least a $0.99$ fraction of positions (in the $j_0$-th block), or $\boxA$'s input was $(A,1)$ and its output agrees with Eve's prediction on a fraction of positions that is between $0.48$ and $0.52$ she claims ``Bod had a 0''. Otherwise she claims ``Bob had a 1''. 

Clearly if Bob is using $(A,0)$ as his input then Alice will predict correctly with probability at least $0.92$, since in that case $G^A_{j_0}$ implies that Eve predicts $\boxB$'s output with at most $1\%$ of error.  If he is using $(B,1)$ then $G^A_{j_0}$ implies that Eve's prediction will be within $0.17$ relative Hamming distance of $\boxB$'s output in block $j_0$. By the CHSH constraint $\boxA$'s output must also be within $0.16$ of $\boxB$'s output, whatever input Alice chooses. Hence $\boxA$'s output is always within $0.43 < 0.49$ of $\boxB$'s, meaning Alice will correctly claim Bob had a $1$ whenever her input is $(A,1)$. Hence in that case she correctly predicts Bob's input with probability at least $0.92/2$.

Overall, conditioned on Alice not aborting her prediction is correct with probability at least $0.69$ over the choice of a random input for Bob, indicating a violation of the no-signaling assumption on the boxes and proving Theorem~\ref{thm:main-quant}. 

\paragraph{Acknowledgments.} We thank Matthew Coudron for useful comments on a preliminary version of this manuscript. 

\bibliography{randomness}

\appendix

\section{Identifying ``good'' blocks in Protocols~A and~B}\label{app:special}

In this section we prove Claim~\ref{claim:special-it} and Claim~\ref{claim:eve-pred}, which play an analogous role for Theorem~\ref{thm:main-class} and Theorem~\ref{thm:main-quant} respectively: that of identifying a special iteration of the protocol that will be useful to Alice and Bob in the guessing game. 

\begin{proof}[Proof of Claim~\ref{claim:special-it}]
 Let $\BAD'$ be the set of strings $b\in\big(\{0,1\}^k\big)^m$ such that $\Pr(b|\CHSH) > 2^{-n}$. Assumption (i) together with Claim~\ref{claim:smoothcap} show that $\Pr(\BAD'|\CHSH) \geq \eps$, so using (ii) we get $\Pr(\CHSH|\BAD')\geq \eps^2$. Define $\BAD$ to contain only those strings $b\in\BAD'$ such that $\Pr(\CHSH|B=b) \geq \eps^2/2$; we have $\Pr(\BAD) \geq (\eps^2/2)\Pr(\BAD') \geq \eps^4/2$. 

By definition of $\BAD$, using Baye's rule we have that for every $b=(b_1,\ldots,b_m)\in \BAD$,
$$ \Pr(B=b, \CHSH) \,=\, \prod_{i=1}^m \Pr(B_i = b_i, \CHSH_i | \CHSH_{<i}, B_{<i} = b_{<i}) \,\geq\, 2^{-n} \eps^2/2.$$
Taking logarithms on both sides,
$$\sum_{i=1}^m -\log\Pr(B_i = b_i, \CHSH_i | \CHSH_{<i}, B_{<i} = b_{<i}) \,\leq\, n +  3\log (1/\eps)\,\leq\, 2n,$$
assuming as in the statement of the claim that $\eps$ is not too small. By an averaging argument at least $3/4$ of all $i\in [m]$ are such that a fraction at least $3/4$ of all $b\in \BAD$ are such that
\beq\label{eq:spec-1}
\Pr(B_i = b_i, \CHSH_i | \CHSH_{<i}, B_{<i} = b_{<i}) \,\geq\, 2^{-48(n/m)} \,\geq \, 2^{-48/C}.
\eeq
Let $S$ be the set of $i\in [m]$ such that~\eqref{eq:spec-1} holds for a fraction at least $3/4$ of $b\in\BAD$. $S$ is a fixed subset of blocks of size $|S| \geq (3/4)m$.


We apply the same reasoning once more, focusing on the CHSH constraint being satisfied in a Bell block. Let $N$ be a random variable equal to the number of Bell blocks that fall in $S$. Since $S$ is fixed, and each block is chosen to be a Bell block independently with probability $1/\ell$, $N$ is concentrated around $\Delta (|S|/m)\geq \Delta/2$. By a Chernoff bound, the probability that $N$ is less than $\Delta/4$ is at most $e^{-\Delta/16}$, which given our choice of $\Delta$ can be neglected in front of the other events we are considering. For the remainder of the proof we assume that $N\geq C/4$. Let $T_j$ be a random variable denoting the index of the $j$-th Bell block, among those that fall in $S$. Starting from $\Pr(\CHSH|\BAD) \geq \eps^2/2$ and using Baye's rule as before, 
$$\sum_{j=1}^N -\log\Pr(\CHSH_{T_j} | \CHSH_{<T_j},\BAD) \,\leq\, 2\log (1/\eps)+1\,\leq\, 3\log (1/\eps). $$
By an averaging argument and using our assumed lower bound on $N$ this implies that a fraction at least $1/2$ of the Bell blocks in Protocol~A are such that
\beq\label{eq:spec-2b}
\Pr(\CHSH_{T_j} | \CHSH_{<T_j},\BAD) \geq \eps^{24/C}.
\eeq
Let $T_j\in T\cap S$ be a Bell block for which~\eqref{eq:spec-2b} holds. For a fraction at least $\eps^{24/C}/2$ of $b\in\BAD$ it holds that 
\beq\label{eq:spec-2c}
\Pr(\CHSH_{T_j} | \CHSH_{<T_j},B=b) \geq \eps^{24/C}/2.
\eeq
By the union bound, at iteration $T_j$~\eqref{eq:spec-2c} will hold simultaneously with~\eqref{eq:spec-1} for a subset $G$ of $\BAD$ of size at least 
$$ \Pr(G) \,=\, \Pr(G|\BAD)\Pr(\BAD)\,\geq\, \big( \eps^{24/C}/2 - 1/4 \big) \eps^4/2\,\geq\, \eps^5$$
given our choice of parameters. By choosing $C$ large enough,~\eqref{eq:spec-1} implies item~1 in the claim, and~\eqref{eq:spec-2c} implies item~2, given the choice of $\Delta$ made in the claim.
\end{proof}

\bigskip

\begin{proof}[Proof of Claim~\ref{claim:eve-pred}]
By definition, $ \Pr\big( G^B \big) \geq p_s \Pr(\CHSH) \geq p_s \eps$. Conditioned on $G^B$, by Markov's inequality it must be that $d_H(E^B, B) < 0.01$ on a fraction at least $1-100f_e$ of blocks in which the input to $B$ was $(A,0)$. Let $f'_e = 100 f_e$. Let $\eta = 2^{- 10^{-5} f'_e |T|/(2\cdot 100^2)}$, and assume $C$ chosen large enough so that $\eta \leq p_s \eps/6 = \Omega(n^{-8(1+\alpha)})$. This is possible since $|T|$ is sharply concentrated  around $C\log^2\ell$ and $f'_e = \Omega(1/\log \ell)$. 

Among the blocks in which Eve's prediction is correct, nothing distinguishes those Bell blocks in which $\boxB$'s input is $(A,0)$: indeed, we may think of those only being designated as Bell blocks after Eve has made her prediction. By a Chernoff bound the probability that more than a fraction $2f'_e$ of such blocks fall into those for which $G^B_j$ does not hold is upper-bounded by $\eta$. Hence the following holds
\beq
\Pr\big( \Es{j\in T:\,Y_j=(A,0)} \, G^B_j > 1-2f'_e | G^B \big) \,\geq\, 1-\eta.
\eeq
Since $V$ is a fixed subset of $[km]$ of size $|V|=O(m^\gamma\log^2 m)$, the probability that any of the randomly chosen $O(\log^2 \ell)$ Bell blocks intersects it is at most $O(m^{-1+\gamma}\log^4 m ) = O(n^{2-1/\gamma}\log^4 n)$ for large enough $n$. We assume as in the statement of Theorem~\ref{thm:main-quant} that $\gamma$ is chosen large enough so that this is much smaller than (our upper bound on) $\eta$, i.e. $\gamma < 1/(9+8\alpha)$. For the remainder of the proof we will neglect the chance of this happening. 

Conditioning further on $\CHSH$ can only blow-up the error by a factor $1/\Pr(\CHSH|G^B)\leq 1/(p_s\eps)$. 
In that case $G^A=G^B$ (Eve's prediction only depends on the advice bits she is given), so we obtain: 
\beq\label{eq:evepred0b}
\frac{\Pr\big( \Es{j\in T:Y_j=0}\, G^B_j > 1-2f'_e,\CHSH | G^A \big)}{\Pr\big(\CHSH|G^A\big)}\,=\,\Pr\big( \Es{j\in T:\,Y_j=(A,0)}\, G^A_j > 1-2f'_e | G^A, \CHSH \big)     \,\geq\, 1-\eta/(p_s\eps).
\eeq
Suppose Eve makes more than a fraction $5f'_e$ of errors in predicting $\boxA$'s output on those Bell blocks in which its input is $(A,0)$. Some of those will later be randomly chosen by Bob as Bell blocks, and by a Chernoff bound with probability at least $1-\eta$ the input to $\boxB$ will also be $(A,0)$ in at least $40\%$ of those blocks. Whenever this happens, Eve's prediction will be wrong on a total fraction more than $2f'_e$ of $\boxB$'s $(A,0)$-input Bell blocks, contradicting~\eqref{eq:evepred0b}. Indeed, whenever $\CHSH$ holds, if the input to both boxes is $(A,0)$ then Eve being correct in predicting $\boxB$'s output is equivalent to her being correct in predicting $\boxA$'s output. Hence the following holds:
\begin{align}
\Pr\big( \Es{j\in T:\,X_j=(A,0)} \,G^A_j > 1-5f'_e ,\CHSH | G^A \big) &\geq \Pr\big( \Es{j\in T:\,Y_j=(A,0)}\, G^A_j > 1-2f'_e,\CHSH | G^A \big)-\eta\notag\\
&\geq (1-\eta/(p_s\eps)) \Pr\big(\CHSH|G^A\big) - \eta \notag\\
&\geq (1-2\eta/(p_s\eps)) \Pr\big(\CHSH|G^A\big),\label{eq:evepred1}
\end{align}
where the last inequality uses $\Pr(\CHSH|G^A)\geq p_s\eps$. As before, since $G^A\wedge \CHSH = G^B \wedge \CHSH$,~\eqref{eq:evepred1} implies the following: 
\beq\label{eq:evepred1b}
\Pr\big( \Es{j\in T:\,X_j=(A,0)} \,G^B_j > 1-5f'_e | G^B,\CHSH \big) \,\geq\, 1-2\eta/(p_s\eps).
\eeq
Next, suppose Eve makes a prediction that is wrong on a fraction at least $14f'_e$ of the Bell blocks, irrespective of Bob's inputs. Then again with high probability at least $40\%$ of the inputs to $\boxA$ in those blocks will be $(A,0)$, implying that Eve is wrong on more than a fraction $5f'_e$ of $\boxA$'s $(A,0)$ inputs, and contradicting~\eqref{eq:evepred1b}. Hence the following is proven just as~\eqref{eq:evepred1} was:
\beq\label{eq:evepred2}
\Pr\big( \Es{j\in T}\, G^B_j > 1- 14f'_e| G^B, \CHSH \big) \,\geq\, 1-3\eta/(p_s\eps).
\eeq
Hence
$$\Pr\big( \Es{j\in T}\, G^A_j > 1-14f'_e | G^A, \CHSH \big) \,\geq\, 1-3\eta/(p_s\eps) ,$$
which is greater than $1/2$ given our choice of $\eta$. 
Removing all conditioning, whenever Eve is given advice bits by Alice, it holds that
$$ \Pr\big( \Es{j\in T}\, G^A_j > 1-14f'_e , \CHSH\big) \,\geq\, \Omega(p_s\eps).$$
\end{proof}


\section{Proof of Lemma~\ref{lem:chain-rule}}\label{app:kxor}

In this appendix we give the proof of Lemma~\ref{lem:chain-rule}. The proof crucially uses properties of a specific extractor construction, first shown to be secure in the presence of quantum bounded-storage adversaries in~\cite{Ta09}, and in the more general setting of quantum bounded-information adversaries in~\cite{DVPR11}. We first describe the extractor. 

\subsection{The $t$-XOR extractor}

The $t$-XOR extractor $E_t$, parametrized by an integer $t$, follows Trevisan's general extractor construction paradigm~\cite{Trevisan01}. It is based on two main ingredients, the $t$-XOR code and a combinatorial design construction due to Hartman and Raz~\cite{HR01}. For us, only the details of the $t$-XOR code will be important. 

\paragraph{The $t$-XOR code.}
Given integers $m$ and $t\leq m$, let $C_t:\{0,1\}^m \to \{0,1\}^{ {m\choose t}}$ map an $m$-bit string to the string of parities of all subsets of $t$ out of its $m$ bits. Two properties of this encoding will be relevant for us. The first is that it is locally computable: each bit of the code only depends on $t$ bits of the input. The second is that it is approximately list-decodable (we summarize its parameters in Lemma~\ref{lem:listxor} below).  

\paragraph{Combinatorial designs.} Given integers $s,m,r$ and $\rho>0$, a collection of subsets $S_1,\ldots,S_r\subseteq [s]$ is called a $(s,m,r,\rho)$ weak design if for all $i\in [r]$, $|S_r|=m$ and for all $j$, $\sum_{i<j} 2^{|S_i\cap S_j|} \leq \rho(r-1)$. For our purposes it will suffice to note that Hartman and Raz~\cite{HR01} proved the existence of a $(s,m,r,1+\gamma)$ design for every $m$, $0<\gamma<1/2$, $s = O(m^2 \log 1/\gamma)$ and $r>s^{\Omega(\log s)}$. 

\paragraph{The $t$-XOR extractor.} We define the extractor that we will use in the proof of Lemma~\ref{lem:chain-rule}. 

\begin{definition} Let $m,r,t,s$ be given integers such that $t = O(\log m)$ and $s = O(\log^4 n)$. Then $E_t:\{0,1\}^m \times \{0,1\}^s \to \{0,1\}^r$  maps $(x,y)\in\{0,1\}^m\times\{0,1\}^s$ to $C_t(x)_{y_{S_1}},\ldots, C_t(x)_{y_{S_r}}$, where $(S_1,\ldots,S_r)$ is a $(s,t\log m,r,5/4)$ design and $y_{S_i}$ designates the bits of $y$ indexed by $S_i$, interpreted as a $t$-element subset of $[m]$. 
\end{definition}

While, as shown in Corollary~5.11 in~\cite{DVPR11}, $E_t$ is a strong extractor with good parameters, we will not use this fact directly. Rather, we will use specific properties that arise from the ``reconstruction paradigm''-based \emph{proof} that it is an extractor secure against quantum adversaries, and one may argue that Lemma~\ref{lem:chain-rule} is implicit in the proof of security of $E_t$ given in~\cite{DVPR11}. Since it does not follow directly from the mere statement that $E_t$ is an extractor, we give more details here. We will show the following lemma, which is more general than Lemma~\ref{lem:chain-rule}. 

\begin{lemma}\label{lem:ext_adv} Let $m,r,t$ be integers such that $t=O(\log^2 m)$ and $\eps>0$. 
Let $\rho_{XE}$ be a cq-state such that $X$ is a random variable distributed over $m$-bit strings. Let $U_r$ be uniformly distributed over $r$-bit strings, and suppose that 
\beq\label{eq:ass-eve}
\|\rho_{Ext(X,Y)E} - \rho_{U_r}\otimes \rho_E \big\|_{tr} \, > \, \eps,
\eeq
i.e. an adversary Eve holding register $E$ can distinguish the output of the extractor from a uniformly random $r$-bit string. Then there exists a fixed subset $V\subseteq [m]$ of size $|V| = O(tr)$ such that, given the string $X_V$ as advice, with probability at least $\Omega(\eps^2/r^2)$ over the choice of $x\sim p_X$ and her own randomness Eve can output a list of $\ell= O(r^4/\eps^4)$ strings $\tilde{x}^1,\ldots,\tilde{x}^\ell$ such that there is an $i\in [\ell]$, $d_H(\tilde{x}^i , x) \leq (2/t)\ln(4r/\eps) $. 
\end{lemma}

It is not hard to see why Lemma~\ref{lem:ext_adv} implies Lemma~\ref{lem:chain-rule}. First note that if $r$ is chosen in Lemma~\ref{lem:ext_adv} so that $r > 2H_\infty^\eps(X|E)$ then the assumption~\eqref{eq:ass-eve} is automatically satisfied.\footnote{The extra randomness coming from the seed of the extractor will be small, as its size can be taken to be $s = O(\log^4 m)$.} The conclusion of Lemma~\ref{lem:chain-rule} then follows from that of Lemma~\ref{lem:ext_adv} by having Eve output a random string out of her $\ell$ predictions, and choosing $t = \Omega(\log^2 m)$ to ensure that  $(2/t)\ln(4r/\eps) \leq 1/\log m$. 

\medskip

In the remainder of this section we sketch the proof of Lemma~\ref{lem:ext_adv}. The first step, explained in Section~\ref{sec:hybrid}, consists in using a hybrid argument to show that, given~\eqref{eq:ass-eve}, Eve can predict a random $t$-XOR of $X$'s bits with reasonable success probability, given sufficiently many ``advice bits'' about $X$. In the second step, detailed in Section~\ref{sec:kt}, we show using an argument due to Koenig and Terhal~\cite{KT07} that this implies the adversary can in fact recover most $t$-XORs of $X$, simultaneously. Finally, in Section~\ref{sec:list} we use the list-decoding properties of the XOR code to show that as a consequence the adversary can with good probability produce a string that agree with $X$ on a large fraction of coordinates.  

\subsection{The hybrid argument}\label{sec:hybrid}

Suppose that~\eqref{eq:ass-eve} holds. Proposition~4.4 from~\cite{DVPR11} shows that a standard hybrid argument, together with properties of Trevisan's extractor (specifically the use of the seed through combinatorial designs), can be used to show the following claim.

\begin{claim}\label{claim:ext_adv-1} There exists a subset $V\subseteq [m]$ of size $|V|=O( tr)$ such that, given the bits $X_V$, Eve can predict a random $t$-XOR of the bits of $X$ with advantage $\eps/r$. Formally, 
\beq\label{eq:eve-1}
\big\| \rho_{C_t(X)_Y Y V E} - \rho_{U_1}\otimes \rho_{Y} \otimes \rho_{VE} \big\|_{tr} \,>\, \frac{\eps}{r},
\eeq
where $Y$ is a random variable uniformly distributed over $\big[{m \choose t}\big]$ and $V$ is a register containing the bits of $X$ indexed by $V$. 
\end{claim}

\subsection{Recovering all $t$-XORs.}\label{sec:kt}

The next step in the proof of Lemma~\ref{lem:ext_adv} is to argue that Eq.~\eqref{eq:eve-1} implies that an adversary given access to $E'=VE$ can predict not only a random XOR of $X$, but a string $Z$ of length ${m\choose t}$ such that $Z$ agrees with the string $C_t(X)$ of all $t$-XOR's of $X$ in a significant fraction of positions. Classically this is trivial, as one can just repeat the single-bit prediction procedure guaranteed by~\eqref{eq:eve-1} for all possible choices $Y$ of the $t$ bits whose parity one is trying to compute. In the quantum setting it is more tricky. We will follow an argument from~\cite{KT07} showing that~\eqref{eq:eve-1} implies that there is a single measurement, independent of $Y$, that one can perform on $E$ and using the (classical) result of which one can predict the bits $C_t(X)_Y$ with good success on average (over the measurement's outcome and the choice of $Y$). 

\begin{claim}\label{claim:ext_adv-2} Suppose~\eqref{eq:eve-1} holds. Then there exists a measurement $\mathcal{F}$, with outcomes in $\{0,1\}^m$, such that 
\beq\label{eq:eve-5}
\Pr_{x\sim p_X,\,y\sim U_{t\log m}}\big(\, C_t(x)_Y \,= \,C_t(\mathcal{F}(VE))_y\,\big) \geq \frac{1}{2} + \frac{\eps^2}{4r^2}\, ,
\eeq
where $\mathcal{F}(VE)$ denotes the outcome of $\mathcal{F}$ when performed on the cq-state $\rho_{VE}$.
\end{claim}

\begin{proof} Our argument closely follows the proof of Theorem~III.1 from~\cite{KT07}. Given an arbitrary cq-state $\rho_{ZQ}$, define the non-uniformity of $Z$ given $Q$ as
$$d(Z\leftarrow Q) \,:=\, \big\| \rho_{ZQ} - \rho_{U_z} \otimes \rho_Q \big\|_{tr}.$$
Let $\rho_x$ denote the state contained in registers $VE$, conditioned on $X=x$. For a fixed string $y$, define two states
$$\rho_0^y \,:=\, \sum_{x:\, C_t(x)_y=0} \,p_X(x)\, \rho_x \qquad\text{and} \qquad \rho_1^y \,:=\, \sum_{x:\, C_t(x)_y=1} \,p_X(x)\, \rho_x.$$
Then, by definition $d\big(C_t(X)_y \leftarrow VE \big) = \big\| \rho_0^y - \rho_1^y \big\|_{tr}$ is the adversary's maximum success probability in distinguishing those states $\rho_x$ which correspond to an XOR of $0$ from those which correspond to an XOR of $1$. Let $\mathcal{E}_y = \big\{E_y^0,E_y^1\big\}$ be the pretty good measurement corresponding to the pair of states $\big\{\rho_0^y,\rho_1^y\big\}$: 
$$E_y^0 \,=\, \rho_{VE}^{-1/2}\rho_0^y\,\rho_{VE}^{-1/2}\qquad\text{and}\qquad E_y^1\,=\, \rho_{VE}^{-1/2}\rho_1^y\,\rho_{VE}^{-1/2},$$
where $\rho_{VE} = \sum_x P_X(x) \rho_x$. Lemma~2 from~\cite{KT07} (more precisely, Eq.~(19)), shows that the following holds as a consequence of~\eqref{eq:eve-1}: 
\beq\label{eq:eve-2}
\sqrt{ \Es{y}\big[ \,2\,d\big(C_t(X)_y \leftarrow \mathcal{E}^y(VE) \big)\big]} + d(C_t(X)_Y\leftarrow Y) \,>\, \frac{\eps}{r},
\eeq
where $\mathcal{E}^y(VE)$ is the result of the POVM $\mathcal{E}^y$ applied on $\rho_{VE}$, and $ d(C_t(X)_Y\leftarrow Y)$ is the distance from uniform of the one-bit extractor's output, in the absence of the adversary. We may as well assume this term to be small: indeed, if it is more than $\eps/(2r)$ then~\eqref{eq:eve-5} is proved without even having to resort to the quantum system $E$. Hence~\eqref{eq:eve-2} implies
$$
 \Exs{y}{ \,d\big(C_t(X)_y \leftarrow \mathcal{E}_{pgm}^y(VE) \big) \,}\,>\, \frac{\eps^2}{2r^2}\, ,
$$
which can be equivalently re-written as 
\beq\label{eq:eve-3b}
 \Exs{y}{ \, \Tr\big(E_y^0 \,\rho_y^0 \big) + \Tr\big(E_y^1 \,\rho_y^1 \big) \,}\,>\, \frac{1}{2}+\frac{\eps^2}{4 r^2}\, .
\eeq
Following the argument in~\cite{KT07}, we define a new PGM $\mathcal{F}$ with outcomes in $\{0,1\}^{m}$ and POVM elements $F^x = P_X(x) \rho_{VE}^{-1/2} \rho_x \,\rho_{VE}^{-1/2}$. The important point to notice is that for $z\in\{0,1\}$ we have $E_y^z = \sum_{x:\, C_t(x)_y=z} F^x$, hence~\eqref{eq:eve-3b} can be re-written as
$$
 \Es{y}\bigg[ \, \sum_{b:\, C_t(b)_y=0} \,\Tr\big(F^x \,\rho_y^0 \big) + \sum_{b:\, C_t(b)_y=1} \,\Tr\big(F^x \,\rho_y^1 \big) \,\bigg]\,>\, \frac{1}{2}+\frac{\eps^2}{4 r^2}\, ,
$$
which is exactly~\eqref{eq:eve-5}. 
\end{proof}

\subsection{List-decoding the XOR code.}\label{sec:list}

The following lemma (for a reference, see \cite{IJK06}, Lemma 42) states the list-decoding properties of the $t$-XOR code $C_t$ that are important for us.  

\begin{lemma}\label{lem:listxor}
For every $\eta>2t^2/2^m$ and $z \in (\{0,1\}^m)^t$, there is a list of $\ell \leq 4/\eta^2$ elements $x^1,\ldots,x^\ell \in \{0,1\}^m$ such that the following holds:
for every $z' \in \{0,1\}^m$ which satisfies 
$$  \Pr_{\{y_1,\ldots,y_t\} \in \binom{m}{t}} [z_{(y_1,\ldots,y_t)} = \oplus_{i=1}^t z'_{y_i} ] \geq \frac{1}{2} + \eta, $$
  there is an $i \in [\ell]$ such that $$\Pr_{y \sim \mathcal{U}_N} [x^i_y = z'_y] \geq 1-\delta,$$
with $\delta = (1/t) \ln (2/\eta)$. 
\end{lemma}

Claim~\ref{claim:ext_adv-2} implies that, with probability at least $\eps^2/(8r^2)$ over the choice of $x$ and over Eve's own randomness, when measuring her system with $\mathcal{F}$ she will obtain a string $\tilde{z}$ whose $t$-XORs agree with those of $x$ with probability at least $1/2+\eps^2/(8r^2)$. Lemma~\ref{lem:listxor} shows that in that case she can recover a list of at most  $2^8 r^4/\eps^4$ ``candidate'' strings $\tilde{z}^i$ such that there exists at least one of these strings which agrees with $x$ at a (possibly adversarial) fraction $1-\delta$ of positions, where $\delta = (2/t)\ln(4r/\eps)$ given our choice of parameters. Hence Lemma~\ref{lem:ext_adv} is proved.

\end{document}